\documentclass[11pt]{article}

\newcommand{\version}{long}

\usepackage{float}
\makeatletter
\let\@tmp\@xfloat     
\usepackage{fixltx2e}
\let\@xfloat\@tmp                    
\makeatother

\usepackage{algorithm,algpseudocode,amsmath,amssymb,amsthm,complexity,etoolbox,fullpage,graphicx,hyperref,ifthen,mathtools,url,thmtools,thm-restate}
\usepackage[inline,shortlabels]{enumitem}

\usepackage{complexity}

\newclass{\ioPSPACE}{i.o.\text{-}PSPACE}
\newlang{\Halt}{Halt}
\ifundef{\toct}{\usepackage{etoolbox,ifthen,tabulary}

\newcommand{\ifft}{\ifundef{\shortiff}{if and only if }{iff }}
\newcommand{\spc}{{ }}

\newcommand{\inpuborpriv}[2]{\ifdef{\priv}{#2}{#1}}
\newcommand{\inpub}[1]{\inpuborpriv{#1}{}}
\newcommand{\inpriv}[1]{\inpuborpriv{}{#1}}
\newcommand{\indraft}[1]{\ifthenelse{\equal{\version}{draft}}{#1}{}}
\newcommand{\infinal}[1]{\ifthenelse{\equal{\version}{final}}{#1}{Only shown in the final version}}
\newcommand{\inshort}[1]{\ifthenelse{\equal{\version}{short}}{#1}{}}
\newcommand{\inlong}[1]{\ifthenelse{\equal{\version}{long}}{#1}{}}
\newcommand{\inshortorlong}[2]{\inshort{#1}\inlong{#2}}

\newcommand{\intoctornot}[2]{\ifdef{\toct}{#1}{#2}}
\newcommand{\intoct}[1]{\intoctornot{#1}{}}
\newcommand{\notintoct}[1]{\intoctornot{}{#1}}

\newcommand{\classorcat}{\inpuborpriv{class}{category}}

\newcommand{\appref}[1]{Appendix~\ref{app:#1}}
\newcommand{\secref}[1]{Section~\ref{sec:#1}}

\newcommand{\renameenv}[2]{
  \expandafter\let\csname #1#2\expandafter\endcsname
  \csname #1\endcsname
  \expandafter\let\csname end#1#2\expandafter\endcsname
  \csname end#1\endcsname
  \expandafter\let\csname #2\endcsname\relax
  \expandafter\let\csname end#2\endcsname\relax}

\ifundef{\defaultlists}{
  \usepackage[inline,shortlabels]{enumitem}
  \setenumerate[1]{(a),itemsep=0pt,topsep=3pt,parsep=0pt,partopsep=0pt}
  \setenumerate[2]{(i),noitemsep,topsep=3pt,parsep=0pt,partopsep=0pt}
  \setenumerate[3]{(A),noitemsep,topsep=3pt,parsep=0pt,partopsep=0pt}
  \setenumerate[4]{(I),noitemsep,topsep=3pt,parsep=0pt,partopsep=0pt}
  \setitemize{noitemsep,topsep=3pt,parsep=0pt,partopsep=0pt}
  \setdescription{noitemsep,topsep=3pt,parsep=0pt,partopsep=0pt}
  \setlist{noitemsep,topsep=3pt,parsep=0pt,partopsep=0pt}}{}

\newcolumntype{x}[1]{>{\centering\arraybackslash}m{#1}}
}{}{} 
\usepackage{algpseudocode,amsmath,amssymb,etoolbox,fixltx2e,mathtools}

\let\eqref\relax

\DeclareFontFamily{U}{mathx}{\hyphenchar\font45}
\DeclareFontShape{U}{mathx}{m}{n}{
      <5> <6> <7> <8> <9> <10>
      <10.95> <12> <14.4> <17.28> <20.74> <24.88>
      mathx10
      }{}
\DeclareSymbolFont{mathx}{U}{mathx}{m}{n}
\DeclareMathSymbol{\bigtimes}{1}{mathx}{"91}

\algnewcommand{\Input}{\item[\textbf{Input:}]}
\algnewcommand{\Output}{\item[\textbf{Output:}]}
\MakeRobust{\Call} 

\newcommand\restr[2]{{
  \left.\kern-\nulldelimiterspace
  #1
  \vphantom{\big|}
  \right|_{#2}
  }}

\newcommand{\grp}{\mathbf{Grp}}
\newcommand{\graph}{\mathbf{Graph}}

\newcommand{\id}{\mathrm{id}}
\newcommand{\concat}{:}

\newcommand{\abs}[1]{\left\vert#1\right\vert}
\newcommand{\can}{\mathrm{Can}}

\newcommand{\cay}{\mathrm{Cay}}

\newcommand{\iso}{\mathrm{Iso}}

\newcommand{\setb}[2]{\left\{#1 \;\middle|\; #2\right\}}

\newcommand{\nth}[1]{\ensuremath{{#1}^{\mathrm{th}}}}

\newcommand{\figref}[1]{Figure~\ref{fig:#1}}

\newcommand{\thmref}[1]{Theorem~\ref{thm:#1}}
\newcommand{\lemref}[1]{Lemma~\ref{lem:#1}}

\newcommand{\defref}[1]{Definition~\ref{defn:#1}}

\newcommand{\corref}[1]{Corollary~\ref{cor:#1}}

\newcommand{\eqref}[1]{(\ref{eq:#1})}

\newcommand{\eps}{\epsilon}

\newcommand{\bmg}{\mathbf{g}}
\newcommand{\bmh}{\mathbf{h}}

\newcommand{\bmk}{\mathbf{k}}

\newcommand{\la}{\leftarrow}

\newcommand{\ra}{\rightarrow}

\newcommand{\tril}{\triangleleft}

\inshort{\usepackage{amsthm}
\makeatletter
\def\thm@space@setup{\thm@preskip=3pt \thm@postskip=3pt}
\makeatother

\makeatletter
\renewenvironment{proof}[1][\proofname]{\par
  \pushQED{\qed}
  \normalfont
  \topsep3pt \partopsep0pt 
  \trivlist
  \item[\hskip\labelsep
        \itshape
    #1\@addpunct{.}]\ignorespaces
  }{
    \popQED\endtrivlist\@endpefalse
    \addvspace{0pt plus 0pt} 
  }
\makeatother
} 

\usepackage{etoolbox,ifthen,thmtools,thm-restate}

\ifundef{\dontnumberwithin}{\declaretheorem[numberwithin=section]{dummy}}{\declaretheorem{dummy}} 
\declaretheorem[sibling=dummy]{theorem}

\declaretheorem[sibling=dummy]{lemma}

\declaretheorem[sibling=dummy]{definition}

\declaretheorem[sibling=dummy]{corollary}

\ifundef{\defaultthmcontinues}{\renewcommand{\thmcontinues}[1]{}}{}

\newcommand{\augcomp}{\mathbf{ACP}}
\newcommand{\augcomptree}{\mathbf{ACPTree}}
\newcommand{\alphadecomp}{\ensuremath{\mathbf{\alpha}\text{-}\mathbf{Decomp}}}
\newcommand{\alphapair}{\ensuremath{\mathbf{\alpha}\text{-}\mathbf{Pair}}}
\newcommand{\bigleafprod}{\bigodot}
\newcommand{\leafprod}{\odot}

\inshort{\newcommand{\shortiff}{}}
\inshort{\usepackage[compact]{titlesec}}

\begin{document}

\title{Beating the Generator-Enumeration Bound for Solvable-Group Isomorphism\footnote{A preliminary version of this work appeared as a portion of~\cite{rosenbaum2013a}.}}
\author{David J. Rosenbaum \\ {\small University of Washington} \\ {\small Department of Computer Science \& Engineering} \\ {\small Email: djr@cs.washington.edu}}
\date{December  1, 2014}

\maketitle
\thispagestyle{empty}

\begin{abstract}
  \begin{sloppypar}
    We consider the isomorphism problem for groups specified by their multiplication tables.  Until recently, the best published bound for the worst-case was achieved by the $n^{\log_p n + O(1)}$ generator-enumeration algorithm.  In previous work with Fabian Wagner, we showed an $n^{(1 / 2) \log_p n + O(\log n / \log \log n)}$ time algorithm for testing isomorphism of $p$-groups by building graphs with degree bounded by $p + O(1)$ that represent composition series for the groups and applying Luks' algorithm for testing isomorphism of bounded degree graphs.

    In this work, we extend this improvement to the more general class of solvable groups to obtain an $n^{(1 / 2) \log_p n + O(\log n / \log \log n)}$ time algorithm.  In the case of solvable groups, the composition factors can be large which prevents previous methods from outperforming the generator-enumeration algorithm.  Using Hall's theory of Sylow bases, we define a new object that generalizes the notion of a composition series with small factors but exists even when the composition factors are large.  By constructing graphs that represent these objects and running Luks' algorithm, we obtain our algorithm for solvable-group isomorphism.  We also extend our algorithm to compute canonical forms of solvable groups while retaining the same complexity.  
  \end{sloppypar}
\end{abstract}


\inlong{\notintoct{
    \newpage
    \setcounter{page}{1}}}

\section{Introduction}
\label{sec:intro}
We study the \emph{group isomorphism problem} in which we must decide if two finite groups given as Cayley tables are isomorphic.

While efficient algorithms are known for testing isomorphism of various special types of groups~\cite{lipton1977a,savage1980a,vikas1996a,kavitha2007a,legall2008a,qiao2011a,babai2011a,codenotti2011a,babai2012a,babai2012b,grochow2013a}, the class of $p$-groups, and hence the more general solvable groups, are conjectured~\cite{babai2011a,codenotti2011a,babai2012a} to contain the hard case of the group isomorphism problem.  For these classes, the $n^{\log_p n + O(1)}$ \emph{generator-enumeration bound}~\cite{felsch1970a,lipton1977a,miller1978a}, where $p$ is the smallest prime dividing the order of the group, has been the tightest worst-case result for several decades.  Deriving this bound is straightforward: such a group must have a generating set of size at most $\log_p n$; since any isomorphism is defined by the image of any generating set, we can test isomorphism by considering all $n^{\log_p n + O(1)}$ possible images of the set.  Obtaining an upper bound of $n^{(1 - \eps) \log_p n + O(1)}$ where $\eps > 0$ for the class of $p$-groups was therefore a longstanding open problem~\cite{lipton2011a}.

In previous work with Wagner~\cite{rosenbaum2013c} (following~\cite{wagner2011a,rosenbaum2013a}), this was accomplished by showing a square root speedup over generator enumeration for the class of $p$-groups\footnote{Subsequent to~\cite{rosenbaum2013c}, James Wilson (personal communication) showed that the algorithm for $p$-group isomorphism from~\cite{obrien1994a} runs in at most $n^{c \log_p n + O(1)}$ time where $c < 1 / 4$.  However, his analysis has not been published and is limited to $p$-groups.}.  However, this left the problem of obtaining an improvement over the generator-enumeration algorithm for solvable groups unresolved.

In the present paper, we show a similar upper bound for the class of solvable-groups.  Our construction is based on the techniques of~\cite{rosenbaum2013c} combined with Hall's theory of Sylow bases~\cite{hall1938a}.

For purposes of comparison, we review the algorithm of~\cite{rosenbaum2013c}.  We say that two composition series are isomorphic if there is an isomorphism that sends each subgroup in the first series to the corresponding subgroup in the second series.  The algorithm consists of two main steps:
\begin{enumerate}[1)]
\item an $n^{(1 / 2) \log_p n + O(1)}$ time Turing reduction from group isomorphism to composition-series isomorphism and
\item an algorithm for testing $p$-group composition series isomorphism in $n^{O(p)}$ time.
\end{enumerate}
Step (1) follows by bounding the number of composition series.  For step (2), we construct rooted trees whose levels represent the factors in the composition series; the multiplication table is then encoded by attaching gadgets to the leaves.  Since the orders of the composition factors bound the number of children at the corresponding levels of the tree and each leaf is connected to a constant number of gadgets, the resulting graph has degree at most $p + O(1)$.  This yields a polynomial-time Karp reduction from composition-series isomorphism to low-degree graph isomorphism.  Combining this with an $n^{O(d)}$ time algorithm~\cite{luks1982a,babai1983a,babai1983b} for testing isomorphism of graphs of degree at most $d$ yields an $n^{O(p)}$ time algorithm for $p$-group composition-series isomorphism as claimed in step (2).

Combining steps (1) and (2) yields an $n^{(1 / 2) \log_p n + O(p)}$ algorithm for $p$-groups (we will refer to this as the graph-isomorphism component of the $p$-group algorithm).  This algorithm is faster than generator-enumeration when $p$ is small and slower when it is large.  (We consider a prime small if it is at most $\alpha = \log n / \log \log n$ and large if it is greater than $\alpha$.)  By choosing between these two algorithms according to the value of $p$, we obtain an $n^{(1 / 2) \log_p n + O(\log n / \log \log n)}$ time algorithm; this gives a square root speedup over generator enumeration regardless of the value of $p$.

Our main result leverages Hall's theory of Sylow bases~\cite{hall1938a} to extend this algorithm to solvable groups.

\begin{restatable}{theorem}{soliso}
  \label{thm:sol-iso}
  Solvable-group isomorphism is decidable in $n^{(1 / 2) \log_p n + O(\log n / \log \log n)}$ deterministic time.
\end{restatable}

The algorithm for solvable groups follows the same framework but is more complicated.  The main conceptual challenge is because solvable groups can have composition factors of large order as well as other composition factors of small order.  This is problematic since both generator enumeration and the graph-isomorphism based $p$-group algorithm just described will take roughly $n^{\log_p n}$ time for a group that has many small composition factors and one large composition factor.

In order to overcome this obstacle, we need a way to (in effect) apply the graph-isomorphism component of the $p$-group algorithm to the part of the group that corresponds to the small prime factors while applying the generator-enumeration algorithm to the part of the group that corresponds to large prime factors.  Since these two parts of a solvable group do not form a direct product decomposition, we need a way of actually combining these two algorithms since we cannot separate the group into independent parts and run the algorithms separately.

Wagner~\cite{wagner2012a} gave a method for reducing the degree of the graph by restricting the isomorphism to be fixed on the quotient of $G$ by a subgroup $G_i$ in the composition series.  If there is a subgroup $G_i$ in the composition series whose prime divisors are all large, then the number of ways of fixing the isomorphism on the quotient $G / G_i$ is relatively small so we can test isomorphism of the composition series.  Thus, we could handle large composition factors if we had a way of moving all the large primes to the top of the composition series.

Since it is not clear that there is always a composition series with all the large primes at the top, we use a different structure.  The key idea in our algorithm for solvable-group isomorphism is to use Sylow bases to separate the large and small prime divisors \footnote{We thank Laci Babai for suggesting this simplification.  An earlier version of this work broke $G$ into many factors which made it more complicated.} (according to the threshold $\alpha = \log n / \log \log n$) into subgroups $P_1$ and $P_2$ of $G$ such that $G = P_1 P_2$.  We call the pair $(P_1, P_2)$ an \emph{$\alpha$-decomposition} for $G$ and define it formally later.  We also let $(Q_1, Q_2)$ be an $\alpha$-decomposition for $H$.  The correctness of this step is guaranteed by the following lemma which follows easily from Hall's theorems~\cite{hall1938a}.


\begin{restatable}{lemma}{solred}
  \label{lem:sol-red}
  For any $\alpha$, solvable-group isomorphism is deterministic polynomial-time Turing-reducible to testing isomorphism of $\alpha$-decompositions of the group.
\end{restatable}

 We then choose a composition series $S_2$ for $P_2$ and a composition series $S_2'$ for $Q_2$.  There is no need to choose composition series for $P_1$ and $Q_1$ since we plan to apply Wagner's degree reduction trick to these subgroups.  We call the pairs $(P_1, S_2)$ and $(Q_1, S_2')$ \emph{$\alpha$-composition pairs} for $G$ and $H$.  We say that $(P_1, S_2)$ is isomorphic to $(Q_1, S_2')$ if there is an isomorphism from $G$ to $H$ that restricts to isomorphisms from $P_1$ to $Q_1$ and $S_2$ to $S_2'$.  By enumerating all possible composition series as in the case for $p$-groups, we can reduce the problem to $\alpha$-composition pair isomorphism.

\begin{restatable}{lemma}{alphared}
  \label{lem:alpha-red}
  Testing isomorphism of the $\alpha$-decompositions $(P_1, P_2)$ and $(Q_1, Q_2)$ of the groups $G$ and $H$ is $n^{(1 / 2) \log_p n + O(1)}$ deterministic time Turing reducible to testing isomorphism of $\alpha$-composition pairs for $(P_1, P_2)$ and $(Q_1, Q_2)$ where $p$ is the smallest prime dividing the order of the group.
\end{restatable}

It remains to show how to test if two $\alpha$-composition pairs are isomorphic.  Solving this problem is the main challenge in generalizing the $p$-group algorithm to solvable groups.  As before, we accomplish this by constructing a graph.  However, now our graph for $G$ must represent both the decomposition $G = P_1 P_2$ and the composition series $S_2$.  We start by constructing a tree; the top of the tree corresponds to the subgroup $P_1$ while the bottom corresponds to $S_2$.  The degree of the top part of the tree is reduced to a constant using Wagner's trick at the cost of a factor of $n^{\alpha + O(1)}$.  Extra gadgets are used to require any isomorphism to respect the decomposition $G = P_1 P_2$.  The multiplication table is represented by attaching gadgets to the leaves in the same way as before.  The result is a graph that has degree $\alpha + O(1)$ and represents the isomorphism class of the $\alpha$-composition pair $(P_1, S_2)$.  Combining with the $n^{O(d)}$ time algorithm~\cite{luks1982a,babai1983a,babai1983b} for testing isomorphism of graphs of degree at most $d$ completes the proof of \thmref{sol-iso}.

As in the case of $p$-groups, we extend our algorithm for solvable-group isomorphism to compute canonical forms of solvable groups within the same amount of time\footnote{In a follow-up paper~\cite{rosenbaum2013b}, we show how to combine this canonization algorithm with a general collision detection framework to reduce the $1 / 2$ in the exponent of \thmref{sol-iso} to $1 / 4$.}.

In \secref{alpha-red}, we reduce solvable-group isomorphism to $\alpha$-decomposition isomorphism and from $\alpha$-decomposition isomorphism to $\alpha$-composition pair isomorphism.  In \secref{graph-red}, we present the reduction from $\alpha$-composition pair isomorphism to low-degree graph isomorphism.  In \secref{sol-algorithms}, we derive our algorithms for solvable-group isomorphism.  

\section{Reducing solvable-group isomorphism to $\alpha$-composition pair isomorphism}
\label{sec:alpha-red}
In this section, we define the notions of $\alpha$-decompositions and $\alpha$-composition pairs and show Turing reductions from solvable-group isomorphism to $\alpha$-decomposition isomorphism and from $\alpha$-decomposition isomorphism to $\alpha$-composition isomorphism.  The first reduction can be done in polynomial time using Hall's theorems~\cite{hall1938a} while the second follows by counting the number of composition series.

From now on, we assume for convenience that the groups $G$ and $H$ have the same order; if this is not the case, then $G$ and $H$ are not isomorphic.  We let $\alpha$ be a parameter that we will later set to $\log n / \log \log n$.  We start with the definition of an $\alpha$-decomposition.

\begin{definition}
  \label{defn:alpha-decomp}
  Let $G$ be a group.  An $\alpha$-decomposition of $G$ is a pair of subgroups $(P_1, P_2)$ such that\inshortorlong{
    \begin{enumerate*}
    \item $G = P_1 P_2$,
    \item every prime dividing $\abs{P_1}$ is greater than $\alpha$ and
    \item every prime dividing $\abs{P_2}$ is at most $\alpha$\inshort{.}
    \end{enumerate*}}{
    \begin{enumerate}
    \item $G = P_1 P_2$,
    \item every prime dividing $\abs{P_1}$ is greater than $\alpha$ and
    \item every prime dividing $\abs{P_2}$ is at most $\alpha$
    \end{enumerate}}
\end{definition}

We say that the $\alpha$-decompositions $(P_1, P_2)$ and $(Q_1, Q_2)$ for the groups $G$ and $H$ are isomorphic if there is an isomorphism $\phi : G \ra H$ such that $\phi[P_i] = Q_i$ for each $i$.  In order to reduce solvable-group isomorphism to $\alpha$-decomposition isomorphism, we now recall two of Hall's theorems.  First, we need to define a Sylow basis.

\begin{definition}
  Let $G$ be a group whose order has the prime factorization $n = \prod_{i = 1}^\ell p_i^{e_i}$.  A Sylow basis for $G$ is a set $\setb{P_i'}{1 \leq i \leq \ell}$ where each $P_i'$ is a Sylow $p_i$-subgroup of $G$ and $P_i' P_j' = P_j' P_i'$ for all $i$ and $j$.
\end{definition}

In a Sylow basis $\setb{P_i'}{1 \leq i \leq \ell}$, we will always assume that each $P_i'$ is a Sylow $p_i$-subgroup of $G$.  We say that the Sylow bases $\setb{P_i}{1 \leq i \leq \ell}$ of $G$ and $\setb{Q_i}{1 \leq i \leq \ell}$ of $H$ are isomorphic if there exists an isomorphism $\phi : G \ra H$ such that $\phi[P_i] = Q_i$ for all $i$.  It is easy to construct an $\alpha$-decomposition from a Sylow basis by letting $P_1$ be the product of the Sylow subgroups that correspond to primes that are \emph{greater} than $\alpha$ and letting $P_2$ be the product of the Sylow subgroups that correspond to primes that are \emph{less} than $\alpha$.  

The following theorem\intoct{ due to Hall~\cite{hall1938a} (cf.~\cite{robinson1996a})} is useful for proving that the reduction from solvable-group isomorphism to $\alpha$-decomposition isomorphism takes polynomial time.

\intoctornot{
\begin{theorem}
  \label{thm:sol-syl}
  A group $G$ is solvable \ifft it has a Sylow basis.
\end{theorem}}{
\begin{theorem}[Hall~\cite{hall1938a}, cf.~\cite{robinson1996a}]
  \label{thm:sol-syl}
  A group $G$ is solvable \ifft it has a Sylow basis.
\end{theorem}
}

\intoctornot{Hall also showed~\cite{hall1938a} (cf.~\cite{robinson1996a}) that t}{T}wo Sylow bases $\setb{P_i'}{1 \leq i \leq \ell}$ and $\setb{Q_i'}{1 \leq i \leq \ell}$ of $G$ are conjugate if there exists $g \in G$ such that for all $i$, $P_i'^g = Q_i'$.

\intoctornot{
\begin{theorem}
  \label{thm:sol-conj}
  Any two Sylow bases of a solvable group are conjugate.
\end{theorem}}{
\begin{theorem}[Hall~\cite{hall1938a}, cf.~\cite{robinson1996a}]
  \label{thm:sol-conj}
  Any two Sylow bases of a solvable group are conjugate.
\end{theorem}
}

Notice that this implies that the group $G$ has at most $n$ Sylow bases.  We also require the ability to compute a Sylow basis of a solvable group.  This was shown by Kantor and Taylor~\cite{kantor1988a} in the setting of permutation groups so it also holds in our case where the group is specified by its Cayley table.

\intoctornot{
\begin{theorem}
  \label{lem:comp-syl-basis-poly}
    A Sylow basis of a solvable group can be computed deterministically in polynomial time.
\end{theorem}}{
\begin{theorem}[Kantor and Taylor~\cite{kantor1988a}]
  \label{lem:comp-syl-basis-poly}
    A Sylow basis of a solvable group can be computed deterministically in polynomial time.
\end{theorem}
}

Armed with these results, it is now easy to reduce solvable-group isomorphism to $\alpha$-decomposition isomorphism.  The following lemma from the introduction explains why our results are restricted to the class of solvable groups.  
\solred*

\begin{proof}
  Let $G$ and $H$ be solvable groups of order $n = \prod_{i = 1}^\ell p_i^{e_i}$.  We compute a Sylow basis $\setb{P_i'}{1 \leq i \leq \ell}$ for $G$.  Define $P_1 = \prod_{i : p_i > \alpha} P_i'$ and $P_2 = \prod_{i : p_i \leq \alpha} P_i'$; this is an $\alpha$-decomposition for $G$.  We compute a Sylow basis $\setb{Q_i'}{1 \leq i \leq \ell}$ for $H$ and consider all of its $n$ conjugates $\setb{Q_i'^h}{1 \leq i \leq \ell}$ where $h \in H$.  For each of these, we define $Q_1 = \prod_{i : p_i < \alpha} Q_i'$ and $Q_2 = \prod_{i : p_i \leq \alpha} Q_i')$ and test if the $\alpha$-decompositions $(P_1, P_2)$ and $(Q_1, Q_2)$ are isomorphic.  We claim that $G \cong H$ \ifft $(P_1, P_2)$ is isomorphic to one of the $(Q_1, Q_2)$ computed above.

  Clearly, if $G$ and $H$ are not isomorphic then no $\alpha$-decomposition of $G$ is isomorphic to an $\alpha$-decomposition of $H$.  If $\phi : G \ra H$ is an isomorphism, then $\setb{\phi[P_i']}{1 \leq i \leq \ell}$ is a Sylow basis for $H$.  By \thmref{sol-conj}, it is equal to some conjugate of $\setb{Q_i'}{1 \leq i \leq \ell}$.  Then

  \begin{equation*}
    (Q_1, Q_2) = (\prod_{i : p_i < \alpha} \phi[P_i'], \prod_{i : p_i \leq \alpha} \phi[P_i'])
  \end{equation*}
is an $\alpha$-decomposition for $H$ that is isomorphic to $(P_1, P_2)$ and our reduction will test if $(P_1, P_2)$ is isomorphic to $(Q_1, Q_2)$.
\end{proof}

Next, we reduce $\alpha$-decomposition isomorphism to $\alpha$-composition pair isomorphism.  First, we define the notion of an $\alpha$-composition pair.

\begin{definition}
  An $\alpha$-composition pair for an $\alpha$-decomposition $(P_1, P_2)$ of a solvable group $G$ is a pair $(P_1, S_2)$ where $S_2$ is a composition series for $P_2$.
\end{definition}

For convenience, we will sometimes say that $(P_1, S_2)$ is an $\alpha$-composition pair for $G$.  Let $(P_1, S_2)$ and  $(Q_1, S_2')$ be a $\alpha$-decompositions for $G$ and $H$.  Then $(P_1, S_2)$ and  $(Q_1, S_2')$ are isomorphic if there is an isomorphism $\phi$ from $(P_1, P_2)$ to $(Q_1, Q_2)$ which restricts to an isomorphism\footnote{Two composition series $G_0 = 1 \tril \cdots \tril G_m = G$ and $H_0 = 1 \tril \cdots \tril H_{m'} = H$ are isomorphic if there exists an isomorphism $\phi : G \ra H$ such that each $\phi[G_i] = H_i$.} from $S_2$ to $S_2'$.

The reduction from $\alpha$-decomposition isomorphism to $\alpha$-composition pair isomorphism, requires an upper bound on the number of composition series for a group and a way to enumerate all composition series.  This was shown in~\cite{rosenbaum2013c}; we repeat it here for the convenience of the reader.  The argument is based on a suggestion by Laci Babai; previously, we used a more complex argument to enumerate all composition series within a particular class.

\intoctornot{
\begin{lemma}
  \label{lem:comp-bound}
  \begin{sloppypar}
    Let $G$ be a group.  Then the number of composition series for $G$ is at most $n^{(1 / 2) \log_p n + O(1)}$ where $p$ is the smallest prime dividing the order of the group.  Moreover, one can enumerate all composition series for $G$ in $n^{(1 / 2) \log_p n + O(1)}$ time.
  \end{sloppypar}
\end{lemma}}{
\begin{lemma}[\cite{rosenbaum2013c}]
  \label{lem:comp-bound}
  \begin{sloppypar}
    Let $G$ be a group.  Then the number of composition series for $G$ is at most $n^{(1 / 2) \log_p n + O(1)}$ where $p$ is the smallest prime dividing the order of the group.  Moreover, one can enumerate all composition series for $G$ in $n^{(1 / 2) \log_p n + O(1)}$ time.
  \end{sloppypar}
\end{lemma}
}

\begin{proof}
  We show that one can enumerate a class of chains that contains all maximal chains of subgroups in $n^{\log_p n + O(1)}$ time.  Since every maximal chain of subgroups contains at most one composition series as a subchain, this suffices to prove the result.

  We start by choosing the first nontrivial subgroup in the series.  Each of these is generated by a single element so there are at most $n$ choices.  If we have a chain $G_0 = 1 < \cdots < G_k$ of subgroups of $G$, then the next subgroup in the chain can be chosen in at most $\abs{G / G_k}$ ways since different representatives of the same coset generate the same subgroup.  Since each $\abs{G_{i + 1}} \geq p \abs{G_i}$, we see that the number of choices $\abs{G / G_k}$ for $G_{k + 1}$ is at most $n / p^k$.  Therefore, the total number of choices required to construct a chain of subgroups in this manner is at most
  \begin{align*}
    \prod_{k = 0}^{\lfloor \log_p n \rfloor - 1} (n / p^k) & \leq p^{\sum_{k = 0}^{\lceil \log_p n \rceil} k} \\
                                                        {} & = p^{(1 / 2) \log_p^2 n + O(\log_p n)} \\
                                                        {} & \leq n^{(1 / 2) \log_p n + O(1)}
  \end{align*}
  
  Since the set of subgroup chains enumerated by this process includes all maximal chains of subgroups, the result follows. 
\end{proof}

We are now ready to derive the algorithm discussed in the introduction.

\alphared*

\begin{proof}
  Let $S_2$ be an arbitrary composition series for $P_2$.  For each composition series $S_2'$ for $Q_2$, we test if the $\alpha$-composition pairs $(P_1, S_2)$ and $(Q_1, S_2')$ are isomorphic.  If $\phi : (P_1, P_2) \ra (Q_1, Q_2)$ is an isomorphism, then $(Q_1, \phi[S_2])$ is an $\alpha$-composition pair for $H$ that is isomorphic to $(P_1, S_2)$.  Thus, the $\alpha$-decompositions $(P_1, P_2)$ and $(Q_1, Q_2)$ are isomorphic \ifft the $\alpha$-composition pair $(P_1, S_2)$ is isomorphic to $(Q_1, S_2')$ for some composition series $S_2'$ for $Q_2$.  The order of $Q_2$ is at most $n$; the smallest prime dividing the order of $Q_2$ is equal to the smallest prime dividing the order of $H$ by \defref{alpha-decomp}.  The complexity then follows from \lemref{comp-bound}.
\end{proof}

We can also prove Turing reductions from solvable-group canonization to $\alpha$-decomposition canonization and from $\alpha$-decomposition canonization to $\alpha$-composition canonization.  For the convenience of the reader, we explicitly define canonical forms of groups, $\alpha$-decompositions and $\alpha$-decomposition pairs.

\begin{definition}
  A map $\can_{\grp}$ is a canonical form for groups if for each group $G$, $\can_{\grp}(G)$ is an $n \times n$ multiplication table with elements in $[n]$ that is isomorphic to $G$, such that, if $G$ and $H$ are groups, $G \cong H$ \ifft $\can_{\grp}(G) = \can_{\grp}(H)$.
\end{definition}

\begin{definition}
  \label{defn:alpha-decomp-can}
  A map $\can_{\alphadecomp}$ is a canonical form for $\alpha$-decompositions if for each $\alpha$-decomposition $(P_1, P_2)$ of a group $G$, $\can_{\alphadecomp}(P_1, P_2) = (M, \psi[P_1], \psi[M_2])$ such that the following hold.

  \begin{enumerate}
  \item $M$ is an $n \times n$ matrix with entries in $[n]$.
  \item $M$ is the multiplication table for a group that is isomorphic to $G$ under the isomorphism $\psi : G \ra [n]$.
  \item
    \begin{sloppypar}
      If $(P_1, P_2)$ and $(Q_1, Q_2)$ are $\alpha$-decompositions then $(P_1, P_2) \cong (Q_1, Q_2)$ \ifft $\can_{\alphadecomp}(P_1, P_2) = \can_{\alphadecomp}(Q_1, Q_2)$.
    \end{sloppypar}
  \end{enumerate}
\end{definition}

\begin{definition}
  \label{defn:alpha-comp-pair-can}
  A map $\can_{\alphapair}$ is a canonical form for $\alpha$-composition pairs if for each $\alpha$-composition pair $(P_1, S_2 = (P_{2, 0} = 1 < \cdots < P_{2, m} = P_2))$ of an $\alpha$-decomposition $(P_1, P_2)$ of a group $G$, $\can_{\alphapair}(P_1, S_2) = (M, \psi[P_1], \psi[P_{2, 0}], \ldots, \psi[P_{2, m}])$ such that the following hold.

  \begin{enumerate}
  \item $M$ is an $n \times n$ matrix with entries in $[n]$.
  \item $M$ is the multiplication table for a group that is isomorphic to $G$ under $\psi : G \ra [n]$.
  \item
    \begin{sloppypar}
      If $(P_1, S_2)$ and $(Q_1, S_2')$ are $\alpha$-decompositions then $(P_1, S_2) \cong (Q_1, S_2')$ \ifft $\can_{\alphapair}(P_1, S_2) = \can_{\alphapair}(Q_1, S_2')$.
    \end{sloppypar}
  \end{enumerate}
\end{definition}

Our canonical form reductions now follow via similar techniques. 

\begin{lemma}
  \label{lem:group-red-alpha-decomp-can}
  Computing the canonical form of a solvable group is polynomial-time Turing reducible to computing canonical forms of $\alpha$-decompositions for the group where $p$ is the smallest prime dividing the order of the group.
\end{lemma}

\begin{proof}
  Let $G$ be a solvable group of order $n = \prod_{i = 1}^\ell p_i^{e_i}$.  For each Sylow basis $\setb{P_i'}{1 \leq i \leq \ell}$ of $G$, we let $P_1 = \prod_{i : p_i > \alpha} P_i'$ and $P_2 = \prod_{i : p_i \leq \alpha} P_i'$ and compute $\can_{\alphadecomp}(P_1, P_2)$.  We define $\can_{\grp}(G)$ to be the multiplication table of the lexicographically least of these canonical forms.  Since two groups are isomorphic \ifft the sets of isomorphism classes of their $\alpha$-decompositions coincide, it follows that $\can_{\grp}$ is a canonical form.  By \thmref{sol-conj}, there are at most $n$ Sylow bases for $G$ which can be enumerated in polynomial time.  Thus, the reduction can be performed in polynomial time.
\end{proof}

\begin{lemma}
  \label{lem:alpha-decomp-red-pair-can}
  Computing the canonical form of an $\alpha$-decomposition of a group is $n^{(1 / 2) \log_p n + O(1)}$ time Turing reducible to computing canonical forms of $\alpha$-composition pairs for the group where $p$ is the smallest prime dividing the order of the group.
\end{lemma}

\begin{proof}
  Let $(P_1, P_2)$ be an $\alpha$-decomposition of a group $G$.  We use \lemref{comp-bound} to enumerate all of the at most $n^{(1 / 2) \log_p n + O(1)}$ composition series $S_2$ for $P_2$.  We define $\can_{\alphadecomp}(P_1, P_2) = (M, \psi[P_1], \psi[P_{2, m}])$ where $(M, \psi[P_1], \psi[P_{2, 0}], \ldots, \psi[P_{2, m}])$ is the lexicographically least canonical form of the $\alpha$-composition pairs $(P_1, S_2)$ that result from this process.  It follows from \defref{alpha-comp-pair-can} that $\can_{\alphadecomp}$ is a canonical form.
\end{proof}

\section{$\alpha$-composition-pair isomorphism and canonization}
\label{sec:graph-red}
In this section, we show our reduction from $\alpha$-composition pair isomorphism to low-degree graph isomorphism.  Our reduction also extends to reducing $\alpha$-composition pair canonization to computing canonical forms of low-degree graphs.  Our proofs follow an outline similar to the analogous reduction from composition series isomorphism to low-degree graph isomorphism in the case of $p$-groups~\cite{rosenbaum2013c}, but are more complex due to the more general structure of solvable groups.

\subsection{Isomorphism testing}
At a high level, our algorithm consists of the following steps.  First, we augment our $\alpha$-composition pair $(P_1, P_2)$ by choosing an ordered generating set $\bmg$ for the subgroup $P_1$ (which corresponds to the large primes) to obtain the \emph{augmented $\alpha$-composition pair} $(P_1, S_2, \bmg)$.  We say that a mapping $\phi : G \ra H$ is an isomorphism between the augmented $\alpha$-decompositions $(P_1, S_2, \bmg)$ and $(Q_1, S_2', \bmh)$ for $G$ and $H$ if $\phi$ is an $\alpha$-composition pair isomorphism for $(P_1, S_2)$ and $(Q_1, S_2')$ and $\phi(\bmg) = \bmh$.  The reason for choosing an augmented $\alpha$-composition pair is so that we can reduce the degree of the part of the graph we construct that corresponds to $P_1$ using the trick due to Wagner~\cite{wagner2011a} mentioned in the introduction.

Since one can fix an ordered generating set $\bmg$ for $P_1$ and consider all possible ordered generating sets for $Q_1$, it is easy to see that $\alpha$-composition pair isomorphism is $n^{\log_{\alpha} n + O(1)}$ Turing-reducible to augmented $\alpha$-composition pair isomorphism.  (Recall that we will later set $\alpha = \log n / \log \log n$ so this is $n^{O(\log n / \log \log n)}$ time and is less than the complexity we are aiming for.)  We state this in the following lemma.

\begin{lemma}
  \label{lem:alpha-comp-red}
  Testing isomorphism of the $\alpha$-composition pairs $(P_1, S_2)$ and $(Q_1, S_2')$ for the solvable groups $G$ and $H$ is $n^{\log_{\alpha} n + O(1)}$ deterministic time Turing reducible to testing isomorphism of augmented $\alpha$-composition pairs for $(P_1, S_2)$ and $(Q_1, S_2')$ where $p$ is the smallest prime dividing the order of the group.
\end{lemma}

We then construct a tree whose leaves represent the elements of $G$; by using the ordered generating set $\bmg$ chosen above, we are able to ensure that the degree of this tree is at most $\alpha + O(1)$.  By augmenting this tree with gadgets that represent the multiplication table of the group, we obtain an object that represents the isomorphism class of the augmented $\alpha$-composition pair $(P_1, P_2, \bmg)$.  The final step of the algorithm is to apply the following result\intoct{ due to Babai and Luks~\cite{babai1983a,babai1983b}} for computing canonical forms of low-degree graphs.

\intoctornot{
\begin{theorem}
  \label{thm:const-deg-can}
  Canonization of colored graphs of degree at most $d$ is in $n^{O(d)}$ time.
\end{theorem}}{
\begin{theorem}[Babai and Luks~\cite{babai1983a,babai1983b}]
  \label{thm:const-deg-can}
  Canonization of colored graphs of degree at most $d$ is in $n^{O(d)}$ time.
\end{theorem}}

The main challenge compared to $p$-group isomorphism~\cite{rosenbaum2013c} is dealing with the fact that some of the prime divisors of a solvable group can be small while others may be large.  This is the main reason why the correctness proof is significantly more complex than for $p$-groups.  Since a $p$-group has exactly one prime divisor, it was possible to handle the cases of small and large primes separately using a graph-isomorphism based $p$-group algorithm~\cite{rosenbaum2013c} (which is fast when the prime is small) and the generator-enumeration algorithm (which is fast when the prime is large).  On the other hand, for solvable groups, it is necessary to design a hybrid algorithm that is fast for both cases simultaneously.

As mentioned above, the first step in the graph construction is to define a tree for an augmented $\alpha$-composition pair $(P_1, P_2, \bmg)$.  We do this by constructing trees $T_1$ and $T_2$ whose leaves correspond to the elements of $P_1$ and $P_2$.  In order to define the part of the tree corresponding to $P_1$, we need a way to canonically order the elements of a group given an ordered generating set.  For completeness, we state and prove the required properties from~\cite{rosenbaum2013c}.

\intoctornot{
\begin{definition}
  \label{defn:gen-ord}
  Let $G$ be a group with an ordered generating set $\bmg = (g_1, \ldots, g_k)$.  Define a total order $\prec_{\bmg}$ on $G$ by $x \prec_{\bmg} y$ if $w_{\bmg}(x) \prec w_{\bmg}(y)$ where each $w_{\bmg}(x) = (x_1, \ldots, x_j)$ is the first word in $\{g_1, \ldots, g_k\}^*$ under the standard ordering such that $x = x_1 \cdots x_j$.
\end{definition}}{
\begin{definition}[\cite{rosenbaum2013c}]
  \label{defn:gen-ord}
  Let $G$ be a group with an ordered generating set $\bmg = (g_1, \ldots, g_k)$.  Define a total order $\prec_{\bmg}$ on $G$ by $x \prec_{\bmg} y$ if $w_{\bmg}(x) \prec w_{\bmg}(y)$ where each $w_{\bmg}(x) = (x_1, \ldots, x_j)$ is the first word in $\{g_1, \ldots, g_k\}^*$ under the standard ordering such that $x = x_1 \cdots x_j$.
\end{definition}
}

\intoctornot{
We will also need the following lemma from~\cite{rosenbaum2013c}.

\begin{lemma}
  \label{lem:gen-ord}
  Let $G$ and $H$ be groups with ordered generating sets $\bmg = (g_1, \ldots, g_k)$ and $\bmh = (h_1, \ldots, h_k)$, and let $x, y \in G$.  Then

  \begin{enumerate}
  \item $\prec_{\bmg}$ is a total ordering on $G$.
  \item if $\phi : G \ra H$ is an isomorphism such that each $\phi(g_i) = h_i$, then $x \prec_{\bmg} y$ \ifft $\phi(x) \prec_{\bmh} \phi(y)$.
  \item we can decide if $x \prec_{\bmg} y$ in $O(n \abs{\bmg})$ time.
  \end{enumerate}
\end{lemma}}{
\begin{lemma}[\cite{rosenbaum2013c}]
  \label{lem:gen-ord}
  Let $G$ and $H$ be groups with ordered generating sets $\bmg = (g_1, \ldots, g_k)$ and $\bmh = (h_1, \ldots, h_k)$, and let $x, y \in G$.  Then

  \begin{enumerate}
  \item $\prec_{\bmg}$ is a total ordering on $G$.
  \item if $\phi : G \ra H$ is an isomorphism such that each $\phi(g_i) = h_i$, then $x \prec_{\bmg} y$ \ifft $\phi(x) \prec_{\bmh} \phi(y)$.
  \item we can decide if $x \prec_{\bmg} y$ in $O(n \abs{\bmg})$ time.
  \end{enumerate}
\end{lemma}
}

\begin{proof}
  Let $S = \{g_1, \ldots, g_k\}$.  For part (a), it is clear that $\prec_{\bmg}$ is a total order since $w_{\bmg} : G \ra S^*$ is clearly injective and the standard ordering on $S^*$ is a total order.

  For part (b), consider an isomorphism $\phi : G \ra H$ such that each $\phi(g_i) = h_i$.  Then if $w_{\bmg}(x) = (x_1, \ldots, x_j)$, $w_{\bmh}(\phi(x)) = (\phi(x_1), \ldots, \phi(x_j))$.  Thus, $x \prec_{\bmg} y$ \ifft $w_{\bmg}(x) \prec w_{\bmg}(y)$ \ifft $w_{\bmh}(\phi(x)) \prec w_{\bmh}(\phi(y))$ \ifft $x \prec_{\bmh} y$.

  For part (c), it suffices to show how to compute $w_{\bmg}(x)$ in polynomial time.  Consider the Cayley graph $\cay(G, S)$ for the group $G$ with generating set $S$.  Then the word $w_{\bmg}(x)$ corresponds to the edges in the minimum length path from $1$ to $x$ in $\cay(G, S)$ that comes first lexicographically.  We can find this path in $O(n \abs{\bmg})$ time by visiting the nodes in breadth-first order starting with $1$.  At the \nth{j} stage, we know $w_{\bmg}(y)$ for all $y \in G$ at a distance of at most $j$ from the root.  We then compute $w_{\bmg}(x)$ for each $x$ at a distance of $j + 1$ from the root by selecting the minimal word $w_{\bmg}(x) \concat g_{x, y}$ over all edges $(x, y)$ associated with an element $g_{x, y}$ of $S$.
\end{proof}

Now we can define the tree that corresponds to $P_1$.  We do this by choosing a balanced binary tree whose leaves are elements of $P_1$.  The choice of this tree is arbitrary so long as it depends only on $\prec_{\bmg}$.  The reason for constructing the trees for $P_1$ and $P_2$ separately is that this allows us to ensure that the tree for $P_1$ has only constant degree.  Otherwise, it would have degree $\Omega(n)$ for groups divisible by large primes which would result in a very slow algorithm.  Later on, we will combine the trees for $P_1$ and $P_2$ to obtain a tree whose leaves correspond to elements of $G$.

\begin{definition}
  \label{defn:TP}
  Let $P_1$ be a group with ordered generating set $\bmg = (g_1, \ldots, g_k)$.  To construct the rooted tree $T(P_1, \bmg)$, we create a leaf node for each element of $P_1$ and color each node by the number that corresponds to its position in the ordering $\prec_{\bmg}$; we then arrange the nodes on a line from smallest to largest according to their colors.  We attach a parent node to each pair of adjacent leaves starting with the smallest pair; if $\abs{P_1}$ is odd, we attach a single parent node to the last leaf.  We then arrange the parent nodes just generated on a line according to the ordering on their children and add new parent nodes for them in the same way.  We continue in this manner until we obtain a single root node from which all the leaves are descended; this yields the tree $T(P_1, \bmg)$.
\end{definition}

Next, we define the tree for the $S_2$ using a definition from~\cite{rosenbaum2013c}.  We start by letting $P_2$ be the root of the tree.  We then partition $P_2$ into the cosets obtained by taking $P_2$ mod the subgroup before $P_2$ in $S_2$.  These are the children of the node $P_2$.  We continue this partitioning process until we obtain cosets in $P_2 / 1$; these correspond to the leaves.  We state the definition for general groups, but in our case the groups will always be solvable.

\intoctornot{
\begin{definition}
  \label{defn:TS}
  Let $P_2$ be a group and consider the composition series $S_2$ given by the subgroups $P_{2, 0} = 1 \tril \cdots \tril P_{2, m} = P_2$.  Then $T(S_2)$ is defined to be the rooted tree whose nodes are $\bigcup_i \left(P_2 / P_{2, i}\right)$.  The root node is $P_2$.  The leaf nodes are $\{x\} \in P_2 / 1$ which we identify with the elements $x \in P_2$.  For each node $x P_{2, i + 1} \in P_2 / P_{2, i + 1}$, there is an edge to each $y P_{2, i}$ such that $y P_{2, i} \subseteq x P_{2, i + 1}$.
\end{definition}}{
\begin{definition}[\cite{rosenbaum2013c}]
  \label{defn:TS}
  Let $P_2$ be a group and consider the composition series $S_2$ given by the subgroups $P_{2, 0} = 1 \tril \cdots \tril P_{2, m} = P_2$.  Then $T(S_2)$ is defined to be the rooted tree whose nodes are $\bigcup_i \left(P_2 / P_{2, i}\right)$.  The root node is $P_2$.  The leaf nodes are $\{x\} \in P_2 / 1$ which we identify with the elements $x \in P_2$.  For each node $x P_{2, i + 1} \in P_2 / P_{2, i + 1}$, there is an edge to each $y P_{2, i}$ such that $y P_{2, i} \subseteq x P_{2, i + 1}$.
\end{definition}
}

In order to obtain a tree whose leaves correspond to elements of $G$, we need to combine the trees for $P_1$ and $S_2$.  For this, we require a variant of the rooted product~\cite{godsil1978a} called a leaf product~\cite{rosenbaum2013c}.  Given two rooted trees, their leaf product is obtained by identifying the root node of a copy of the second tree with each leaf node.

\intoctornot{
\begin{definition}
  \label{defn:leaf-prod}
  Let $T_1$ and $T_2$ be trees rooted at $r_1$ and $r_2$.  Then the leaf product $T_1 \leafprod T_2$ is the tree rooted at $r_1$ with vertex set
  \begin{equation*}
    V(T_1) \cup \setb{(x, y)}{x \in L(T_1) \text{ and } y \in V(T_2) \setminus \{r_2\}}
  \end{equation*}
  The set of edges is
  \begin{align*}
    E(T_1) & \cup \setb{(x, (x, y))}{x \in L(T_1) \text{ and } (r_2, y) \in E(T_2)} \\
    {} & \cup \setb{((x, y), (x, z))}{x \in L(T_1) \text{ and } (y, z) \in E(T_2) \text{ where } y, z \not= r_2}
  \end{align*}
\end{definition}}{
\begin{definition}[\cite{rosenbaum2013c}]
  \label{defn:leaf-prod}
  Let $T_1$ and $T_2$ be trees rooted at $r_1$ and $r_2$.  Then the leaf product $T_1 \leafprod T_2$ is the tree rooted at $r_1$ with vertex set
  \begin{equation*}
    V(T_1) \cup \setb{(x, y)}{x \in L(T_1) \text{ and } y \in V(T_2) \setminus \{r_2\}}
  \end{equation*}
  The set of edges is
  \begin{align*}
    E(T_1) & \cup \setb{(x, (x, y))}{x \in L(T_1) \text{ and } (r_2, y) \in E(T_2)} \\
    {} & \cup \setb{((x, y), (x, z))}{x \in L(T_1) \text{ and } (y, z) \in E(T_2) \text{ where } y, z \not= r_2}
  \end{align*}
\end{definition}
}

For convenience, we identify the tuples $(x, (y, z))$, $((x, y), z)$ with $(x, y, z)$ in the vertex set so that leaf products are associative.  It is also useful to define leaf products of tree isomorphisms and bijections between the leaves of two trees\intoct{ as in~\cite{rosenbaum2013c}}.

\intoctornot{
\begin{definition}
  \label{defn:leaf-prod-iso}
  For each $1 \leq i \leq k$, let $T_i$ and $T_i'$ be trees rooted at $r_i$ and $r_i'$ and let $\phi_i : L(T_i) \ra L(T_i')$ be a bijection that extends to a unique isomorphism which we denote by $\hat \phi : T_i \ra T_i'$.  Then the leaf product $\bigleafprod_{i = 1}^k \phi_i : \bigleafprod_{i = 1}^k T_i \ra \bigleafprod_{i = 1}^k T_i'$ sends each $(x_1, \ldots, x_j)$ to $(\hat \phi_1(x_1), \ldots, \hat \phi_j(x_j))$ where each $x_i \in L(T_i)$ for $i < j$, $x_j \in V(T_j) \setminus \{r_j\}$ and $j \leq k$.
\end{definition}}{
\begin{definition}[\cite{rosenbaum2013c}]
  \label{defn:leaf-prod-iso}
  For each $1 \leq i \leq k$, let $T_i$ and $T_i'$ be trees rooted at $r_i$ and $r_i'$ and let $\phi_i : L(T_i) \ra L(T_i')$ be a bijection that extends to a unique isomorphism which we denote by $\hat \phi : T_i \ra T_i'$.  Then the leaf product $\bigleafprod_{i = 1}^k \phi_i : \bigleafprod_{i = 1}^k T_i \ra \bigleafprod_{i = 1}^k T_i'$ sends each $(x_1, \ldots, x_j)$ to $(\hat \phi_1(x_1), \ldots, \hat \phi_j(x_j))$ where each $x_i \in L(T_i)$ for $i < j$, $x_j \in V(T_j) \setminus \{r_j\}$ and $j \leq k$.
\end{definition}
}

It is easy to see that $\bigleafprod_{i = 1}^k \phi_i$ is a well-defined isomorphism from $\bigleafprod_{i = 1}^k T_i$ to $\bigleafprod_{i = 1}^k T_i'$.  We are now finally in a position to define the tree for a augmented $\alpha$-composition pair.

\begin{definition}
  \label{defn:aug-comp}
  Let $(P_1, S_2, \bmg)$ be an augmented $\alpha$-composition pair for a solvable group $G$.  We define $T(P_1, S_2, \bmg) = T(P_1, \bmg) \leafprod T(S_2)$.
\end{definition}

As in the case of $p$-groups, we cannot attach the aforementioned multiplication gadgets directly to the tree $T(P_1, S_2, \bmg)$ because each leaf be attached to $n$ gadgets and would thus have degree $\Omega(n)$; this would cause our algorithm to be extremely slow.  We resolve this by utilizing the leaf product of $T(P_1, S_2, \bmg)$ with itself so that each multiplication gadget is only attached to a constant number of leaves.

The following notation is convenient as it allows us to easily associate elements of $G$ with nodes in the tree $T(P_1, S_2, \bmg)$.  Let $* : \setb{(x_1, x_2)}{x_i \in P_i} \ra G$ by $*(x_1, x_2) = x_1 x_2$ and note that this is a bijection.  Similarly, we define $\bullet : \setb{(x_1, x_2)}{x_i \in Q_i} \ra H$ by $\bullet(x_1, x_2) = x_1 x_2$.  We can then represent each $x \in G$ by the node $*^{-1}(x)$ in $T(P_1, S_2, \bmg)$ and attach the gadget for each multiplication rule $x y = z$ to the nodes $*^{-1}(x)$, $*^{-1}(y)$ and $*^{-1}(z)$.  We formalize this in the following definition.

\begin{figure}[H]
  \centering
  \includegraphics[scale=1.2]{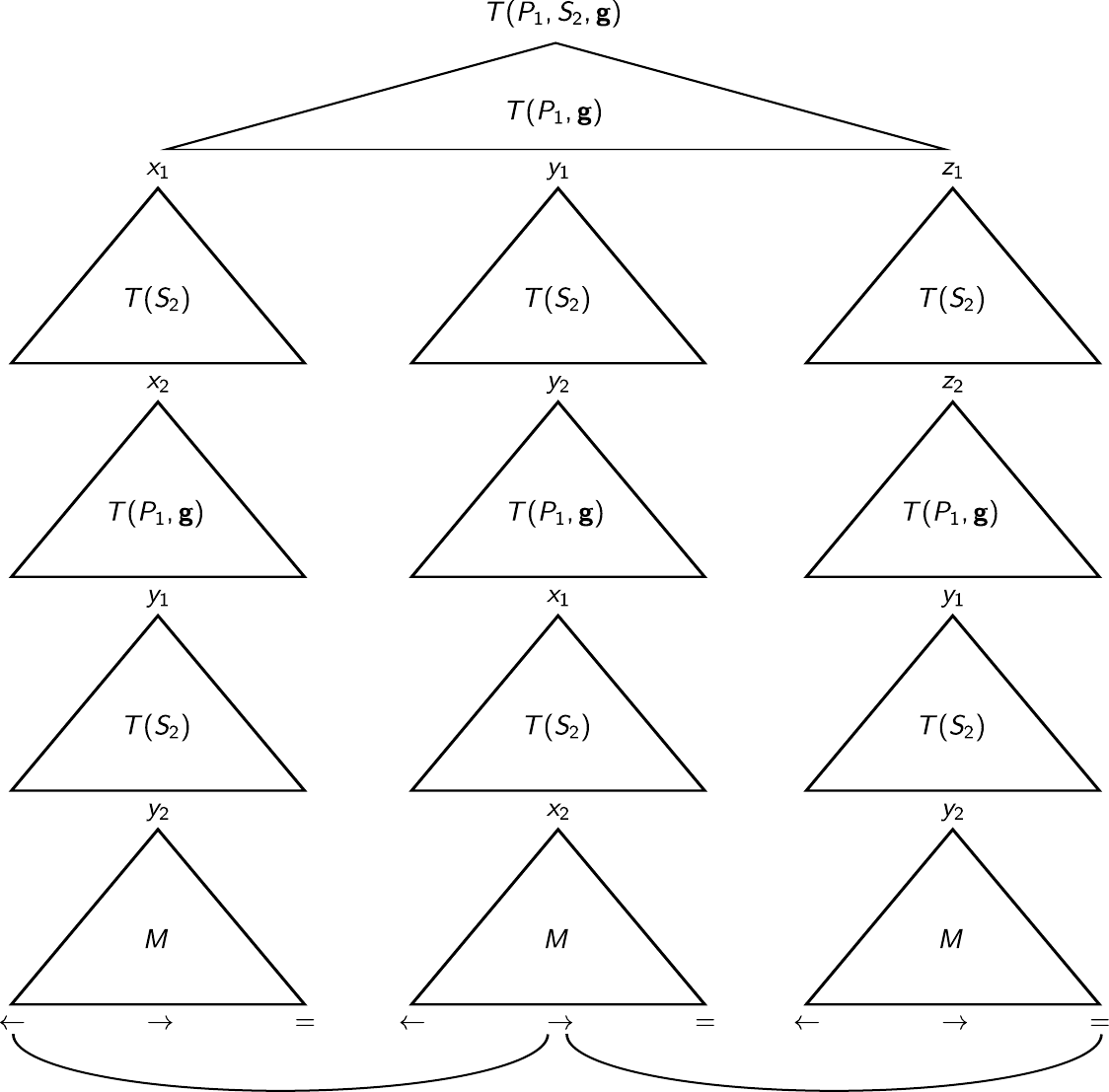}
  \caption{The graph $X(P_1, S_2, \bmg)$ with the multiplication gadget for $x y = z$ where $z = x y$, $*^{-1}(x) = (x_1, x_2)$, $*^{-1}(y) = (y_1, y_2)$ and $*^{-1}(z) = (z_1, z_2)$}
  \label{fig:X-PS}
\end{figure}

\begin{definition}
  \label{defn:X-PS}
  Let $(P_1, S_2, \bmg)$ be an augmented $\alpha$-composition pair for a solvable group $G$ and define $M$ to be the tree with a root connected to three nodes $\la$, $\ra$ and $=$ with colors ``left'', ``right'' and ``equals'' respectively.  We construct $X(P_1, S_2, \bmg)$ by starting with the tree $T(P_1, S_2, \bmg) \leafprod T(P_1, S_2, \bmg) \leafprod M$ and connecting multiplication gadgets to the leaf nodes.  For each $x, y \in G$, we create the path $((*^{-1}(x), *^{-1}(y), \la), (*^{-1}(y), *^{-1}(x), \ra), (*^{-1}(x y), *^{-1}(y), =))$.  We color each node $(x_1, 1)$ where $x_1 \in P_1$ ``second identity.''  Finally, we color the remaining nodes ``internal.''

\end{definition}

The graph $X(P_1, S_2, \bmg)$ can be thought of a rooted tree with edges added between some nodes at the same levels.  The edges from the original tree are called \emph{tree edges} and the edges between nodes at the same level are called \emph{cross edges}.  We show $X(P_1, S_2, \bmg)$ in \figref{X-PS}.

The correctness of our reduction is based on the fact that two augmented composition pairs $(P_1, S_2, \bmg)$ and $(Q_1, S_2', \bmh)$ are isomorphic \ifft $X(P_1, S_2, \bmg)$ and $X(Q_1, S_2', \bmh)$ are isomorphic.  We prove this in the remainder of this subsection.

Some additional terminology is required for the proof.  We define $\augcomp$ to be the \classorcat\spc of augmented composition pairs for finite solvable groups \inpuborpriv{and}{and isomorphisms between them;} let $\augcomptree$ be the \classorcat\spc of graphs that are isomorphic to the graph $X(P_1, S_2, \bmg)$ for some augmented composition pair $(P_1, S_2, \bmg)$\inpriv{ and isomorphisms between such graphs}.  We overload the symbol $X$ from \defref{X-PS} by defining $X(\phi) : X(P_1, S_2, \bmg) \ra X(Q_1, S_2', \bmh)$ to be $\restr{\phi}{P_1} \leafprod \restr{\phi}{P_2} \leafprod \restr{\phi}{P_1} \leafprod \restr{\phi}{P_2} \leafprod \id_M$ for each $\alpha$-composition pair isomorphism $\phi : (P_1, S_2, \bmg) \ra (Q_1, S_2', \bmh)$\inpriv{ (thus obtaining a functor)}.

In order to prove the correctness of our reduction, we need to show that the augmented $\alpha$-composition pairs $(P_1, S_2, \bmg)$ and $(Q_1, S_2', \bmh)$ are isomorphic \ifft the graphs $X(P_1, S_2, \bmg)$ and $X(Q_1, S_2', \bmh)$ are isomorphic.  The forward direction of the implication \inpuborpriv{is equivalent to}{follows from} the assertion that \inpuborpriv{$X_{(P_1, S_2, \bmg), (Q_1, S_2', \bmh)} : \iso((P_1, S_2, \bmg), (Q_1, S_2', \bmh)) \ra \iso(X(P_1, S_2, \bmg), X(Q_1, S_2', \bmh))$ is well-defined}{$X$ is a functor}.  Proving the converse is more difficult and is one of the main lemmas of this subsection.


\begin{lemma}
  \label{lem:X-functor}
  Let $(P_1, S_2, \bmg)$ and $(Q_1, S_2', \bmh)$ be augmented $\alpha$-composition pairs for the solvable groups $G$ and $H$.  Then \inpuborpriv{the map
    \begin{equation*}
      X_{(P_1, S_2, \bmg), (Q_1, S_2', \bmh)} : \iso((P_1, S_2, \bmg), (Q_1, S_2', \bmh)) \ra \iso(X(P_1, S_2, \bmg), X(Q_1, S_2', \bmh))
    \end{equation*}
is well-defined.}{$X : \augcomp \ra \augcomptree$ is a functor.}
\end{lemma}

Before proceeding with the proof, it is convenient to introduce additional notation.  Let $x, y \in G$.  Consider the sequence of nodes that starts at $*^{-1}(x)$, follows tree edges (away from the root) to a node colored ``left'', follows a cross edge to a node colored ``right'', then follows tree edges (towards the root) to $*^{-1}(y)$, follows tree edges (away from the root) back to the same node colored ``right'' and finally follows a cross edge to a node colored ``equal''; we call this a $W$-\emph{sequence} from $x$ to $y$ to $x y$ since its shape resembles a $W$ (see \figref{X-PS}).  Since $W$-sequences correspond to multiplication gadgets, there is exactly one $W$-sequence from $*^{-1}(x)$ to $*^{-1}(y)$: namely, the one that results from the multiplication gadget
\begin{equation*}
  ((*^{-1}(x), *^{-1}(y), \la), (*^{-1}(y), *^{-1}(x), \ra), (*^{-1}(x y), *^{-1}(y), =)).
\end{equation*}
Therefore, we denote \emph{the} $W$-sequence from $x$ to $y$ to $x y$ by $W(x, y)$.  We now proceed with our proof.

\begin{proof}
  Consider the augmented $\alpha$-composition pairs $(P_1, S_2, \bmg)$ and $(Q_1, S_2', \bmh)$ for the solvable groups $G$ and $H$.  Let $\phi : (P_1, S_2, \bmg) \ra (Q_1, S_2', \bmh)$ be an isomorphism and let $P_{2, 0} = 1 \tril \cdots \tril P_{2, m} = P_2$ and $Q_{2, 0} = 1 \tril \cdots \tril Q_{2, m} = Q_2$ be the subgroup chains for $S_2$ and $S_2'$.  Because $\phi(\bmg) = \bmh$, it follows from~\lemref{gen-ord} that $\restr{\phi}{P_1}$ extends to a unique isomorphism between the rooted colored trees $T(P_1, \bmg)$ and $T(Q_1, \bmh)$.  Moreover, since each $\phi[P_{2, i}] = Q_{2, i}$, we see that $\restr{\phi}{P_2}$ extends to a unique isomorphism from $T(S_2)$ to $T(S_2')$.  Thus, $\restr{\phi}{P_1} \leafprod \restr{\phi}{P_2}$ is an isomorphism from $T(P_1, \bmg) \leafprod T(S_2)$ to $T(Q_1, \bmh) \leafprod T(S_2')$; therefore, $X(\phi) = \restr{\phi}{P_1} \leafprod \restr{\phi}{P_2} \leafprod \restr{\phi}{P_1} \leafprod \restr{\phi}{P_2} \leafprod \id_M$ is a tree isomorphism.

  Let $x, y \in G$ and let $*^{-1}(x) = (x_1, x_2)$.  Then $X(\phi)$ maps $*^{-1}(x)$ to $(\phi(x_1), \phi(x_2)) = \bullet^{-1}(\phi(x))$ as $\phi(x) = \phi(x_1) \phi(x_2)$.  Similarly, recalling that we identified expressions of the forms $((x_1, x_2), (y_1, y_2))$ and $(x_1, x_2, y_1, y_2)$, we see that $X(\phi)$ maps $(*^{-1}(x), *^{-1}(y))$ to $(\bullet^{-1}(\phi(x)), \bullet^{-1}(\phi(y)))$

  Consider the path
  \begin{equation*}
    ((*^{-1}(x), *^{-1}(y), \la), (*^{-1}(y), *^{-1}(x), \ra), (*^{-1}(x y), *^{-1}(y), =))
  \end{equation*}
  in $X(P_1, S_2, \bmg)$.  The image of this path under $X(\phi)$ is
  \begin{equation*}
    ((\bullet^{-1}(\phi(x)), \bullet^{-1}(\phi(y)), \la), (\bullet^{-1}(\phi(y)), \bullet^{-1}(\phi(x)), \ra), (\bullet^{-1}(\phi(x y)), \bullet^{-1}(\phi(y)), =)).
  \end{equation*}
  By \defref{X-PS}, this path is one of the multiplication gadgets in $X(Q_1, S_2', \bmh)$.  Thus, $X(\phi)$ maps each $W$-sequence in $X(P_1, S_2, \bmg)$ to a $W$-sequence in $X(Q_1, S_2', \bmh)$.  Moreover,
 $X(\phi)$ maps each node $(x_1, 1)$ to $(\phi(x_1), 1)$, so it respects the ``second identity'' color.  This implies that $X(P_1, S_2, \bmg) \cong X(Q_1, S_2', \bmh)$ since both graphs have the same number of multiplication gadgets (and hence the same number of $W$-sequences).\inpriv{

    Finally, if $(R_1, S_2'', \bmk)$ is an $\alpha$-composition pair and $\psi$ is an isomorphism from $(Q_1, S_2', \bmh)$ to $(R_1, S_2'', \bmk)$, then $X(\psi \phi) = X(\psi) X(\phi)$ and $X(\id_{(P_1, S_2, \bmg)}) = \id_{X(P_1, S_2, \bmg)}$.  Thus, $X$ is a functor.}
\end{proof}

In order to show if that if the graphs $X(P_1, S_2, \bmg)$ and $X(Q_1, S_2', \bmh)$ are isomorphic then so are the augmented $\alpha$-composition pairs $(P_1, S_2, \bmg)$ and $(Q_1, S_2', \bmh)$, it suffices to show that \inpuborpriv{the map $X_{(P_1, S_2, \bmg), (Q_1, S_2', \bmh)} : \iso((P_1, S_2, \bmg), (Q_1, S_2', \bmh)) \ra \iso(X(P_1, S_2, \bmg), X(Q_1, S_2', \bmh))$ is surjective}{$X$ is a full functor}.  This is the key to our correctness proof and implies that augmented $\alpha$-composition pair isomorphism reduces to testing isomorphism of the resulting graphs. 
To do this, we need to show that every isomorphism from $X(P_1, S_2, \bmg)$ to $X(Q_1, S_2', \bmh)$ can be written as a leaf product of group isomorphisms.  We accomplish this by restricting the isomorphism between the graphs to certain subsets of nodes and showing that the isomorphism is the leaf product of these restrictions (which turn out to be group isomorphisms).  An isomorphism $\theta : X(P_1, S_2, \bmg) \ra X(Q_1, S_2', \bmh)$ induces the bijection $\phi = \bullet \circ \theta \circ *^{-1} : G \ra H$.  We call this $\phi$ the \emph{induced bijection} for $\theta$.

\begin{lemma}
  \label{lem:iso-decomp}
  Let $X(P_1, S_2, \bmg)$ and $X(Q_1, S_2', \bmh)$ be augmented $\alpha$-composition pairs for the solvable groups $G$ and $H$, let $\theta : X(P_1, S_2, \bmg) \ra X(Q_1, S_2', \bmh)$ be an isomorphism and let $\phi$ be its induced bijection.  Then


  \begin{enumerate}
  \item $\phi : G \ra H$ is a group isomorphism,
  \item $\phi_1 = \restr{\phi}{P_1} : P_1 \ra Q_1$ and $\phi_2 = \restr{\phi}{P_2} : P_2 \ra Q_2$ are group isomorphisms,
  \item $\theta = \phi_1 \leafprod \phi_2 \leafprod \phi_1 \leafprod \phi_2 \leafprod \id_M$ and
  \item $\phi : (P_1, S_2, \bmg) \ra (Q_1, S_2', \bmh)$ is an augmented $\alpha$-composition pair isomorphism.
  \end{enumerate}
\end{lemma}

\begin{proof}
  Let us start with part (a).  It follows from the assumption that $\theta$ is an isomorphism (and hence bijective) that $\phi$ is a bijection.

  Let $x, y \in G$.  Now, $\theta$ maps the nodes $*^{-1}(x)$ and $*^{-1}(y)$ in $X(P_1, S_2, \bmg)$ to $\bullet^{-1}(\phi(x))$ and $\bullet^{-1}(\phi(y))$ by definition of $\phi$.  It follows that $\theta$ maps the $W$-sequence $W(x, y)$ from $x$ to $y$ to $xy$ in $X(P_1, S_2, \bmg)$ to the $W$-sequence $W(\phi(x), \phi(y))$ in $X(Q_1, S_2', \bmh)$.  Now, since $\theta$ maps $*^{-1}(xy)$ to $\bullet^{-1}(\phi(xy))$, it follows that the $W$-sequence $W(\phi(x), \phi(y))$ in $X(Q_1, S_2', \bmh)$ is from $\phi(x)$ to $\phi(y)$ to $\phi(xy)$.  Therefore, by \defref{X-PS}, $\phi(x y) = \phi(x) \phi(y)$ so $\phi$ is a group isomorphism.


  Now we prove (b).  Let $x_1 \in P_1$.  Because $\theta$ respects the ``second identity'' color, it follows that it maps $(x_1, 1)$ to $(x_1', 1)$ for some $x_1' \in Q_1$.  Then $x_1' = \phi(x_1)$ which implies that $\phi[P_1] = Q_1$.

  Now let $x_2 \in P_2$.  Because $\phi$ is an isomorphism, $\phi(1) = 1$; thus, $\theta$ sends the node $(1, 1)$ to $(1, 1)$ which implies that it maps $1$ to $1$.  Thus, for some $x_2' \in Q_2$,
  \begin{align*}
    \theta(1, x_2)      & = (1, x_2') \\
    \theta(*^{-1}(x_2)) & = \bullet^{-1}(x_2') \\
    \phi(x_2)           & = x_2'.
  \end{align*}
  Thus, $\theta(1, x_2) = (1, \phi(x_2))$ so $\phi[P_2] = Q_2$ and $\phi_2$ is a group isomorphism.

  For part (c), let $x, y \in G$ and $*^{-1}(x) = (x_1, x_2)$.  By part (b), $\theta$ sends the node $x_1$ to $\phi_1(x_1)$.  Therefore, for some $x_2' \in Q_2$,
  \begin{align*}
    \theta(x_1, x_2)          & = (\phi(x_1), x_2') \\
    \bullet(\theta(x_1, x_2)) & = \phi(x_1) x_2' \\
    \phi(x)                   & = \phi(x_1) x_2'.
  \end{align*}
  Since $\phi(x) = \phi(x_1) \phi(x_2)$, this implies that $x_2' = \phi(x_2)$ so $\theta$ maps $*^{-1}(x) = (x_1, x_2)$ to $\bullet^{-1}(\phi(x)) = (\phi(x_1), \phi(x_2))$.
  
  Now consider a node $(*^{-1}(x), *^{-1}(y), \ell)$ where $x, y \in G$ and $\ell \in \{\la, \ra, =\}$.  As $(*^{-1}(x), *^{-1}(y))$ is in the subtree rooted at $*^{-1}(x)$, $\theta$ sends it to a node of the form $(\bullet^{-1}(\phi(x)), \bullet^{-1}(b))$ for some $b \in H$.  Similarly, $\theta$ maps the node $(*^{-1}(y), *^{-1}(x))$ to a node of the form $(\bullet^{-1}(\phi(y)), \bullet^{-1}(a))$ for some $a \in H$.  Now, because $(*^{-1}(x), *^{-1}(y))$ and $(*^{-1}(y), *^{-1}(x))$ are in the $W$-sequence from $x$ to $y$ to $x y$, $(\bullet^{-1}(\phi(x)), \bullet^{-1}(b))$ and $(\bullet^{-1}(\phi(y)), \bullet^{-1}(a))$ are in the $W$-sequence from $\phi(x)$ to $\phi(y)$ to $\phi(x y)$.  Then by \defref{X-PS}, $a = \phi(x)$ and $b = \phi(y)$.  Therefore, $\theta$ maps $(*^{-1}(x), *^{-1}(y))$ to $(*^{-1}(\phi(x)), *^{-1}(\phi(y)))$.  Because of the coloring of the leaves in \defref{X-PS}, it follows that $\theta = \phi_1 \leafprod \phi_2 \leafprod \phi_1 \leafprod \phi_2 \leafprod \id_M$.

  Finally, let us prove part (d).  We already know that $\phi$ is a group isomorphism by part (a).  By part (b), we know that each $\phi[P_i] = Q_i$.

  Let $P_{2, 0} = 1 \tril \cdots \tril P_{2, m} = P_2$ and $Q_{2, 0} = 1 \tril \cdots \tril Q_{2, m} = Q_2$ be the subgroup chains for $S_2$ and $S_2'$.  We need to show that each $\phi[P_{2, i}] = Q_{2, i}$.  By part (c), $\theta$ maps $(1, 1)$ in $X(P_1, S_2, \bmg)$ to $(1, 1)$ in $X(Q_1, S_2', \bmh)$.  Now the path from the root of $X(P_1, S_2, \bmg)$ to $(1, 1)$ contains the nodes $(1, P_{2, m}), \ldots, (1, P_{2, 0})$ (in that order).  Moreover, the descendants of the node $(1, P_{2, i})$ that are in $P_1 \times P_2$ are $\setb{(1, x_2)}{x_2 \in P_{2, i}}$.  Similarly, the path from the root of $X(Q_1, S_2', \bmh)$ to $(1, 1)$ contains the nodes $(1, Q_{2, m}), \ldots, (1, Q_{2, 0})$ (in that order) and the descendants of the node $(1, Q_{2, i})$ that are also in $Q_1 \times Q_2$ are $\setb{(1, x_2')}{x_2' \in Q_{2, i}}$.  Therefore, $\theta$ maps each set $\setb{(1, x_2)}{x_2 \in P_{2, i}}$ to $\setb{(1, x_2')}{x_2' \in Q_{2, i}}$.  Then, by definition of $\phi$, $\phi[P_{2, i}] = Q_{2, i}$ and part (d) is proved.
\end{proof}

We now prove that \inpuborpriv{$X_{(P_1, S_2, \bmg), (Q_1, S_2', \bmh)}$ is bijective}{$X$ is a fully faithful functor}.  For isomorphism testing, we only need to show that it is \inpuborpriv{surjective}{full}.  However, we will need it to be \inpuborpriv{injective}{faithful} later when we discuss canonical forms.

\begin{theorem}
\label{thm:X-fff}
\inpuborpriv{Let $(P_1, S_2, \bmg)$ and $(Q_1, S_2', \bmh)$ be augmented $\alpha$-composition pairs for the solvable groups $G$ and $H$.  Then $X_{(P_1, S_2, \bmg), (Q_1, S_2', \bmh)}$ is a bijection}{$X : \augcomp \ra \augcomptree$ is a fully faithful functor and can be evaluated in polynomial time}.\inpub{  Moreover, both $X(P_1, S_2, \bmg)$ and $X(\phi)$ where $\phi \in \iso((P_1, S_2, \bmg), (Q_1, S_2', \bmh))$ can be computed in polynomial time.}
\end{theorem}

\begin{proof}
We know that \inpuborpriv{$X_{(P_1, S_2, \bmg), (Q_1, S_2', \bmh)}$ is well-defined}{$X$ is a functor} by \lemref{X-functor}.  Let $\theta : X(P_1, S_2, \bmg) \ra X(Q_1, S_2', \bmh)$ be an isomorphism.  By \lemref{iso-decomp}, the induced bijection $\phi : (P_1, S_2, \bmg) \ra (Q_1, S_2', \bmh)$ is an isomorphism and $\theta = \phi_1 \leafprod \phi_2 \leafprod \phi_1 \leafprod \phi_2 \leafprod \id_M$ where each $\phi_i = \restr{\phi}{P_i}$.  Then $X(\phi) = \theta$ so $X_{(P_1, S_2, \bmg), (Q_1, S_2', \bmh)}$ is surjective.

Let $\phi, \psi : (P_1, S_2, \bmg) \ra (Q_1, S_2', \bmh)$ be isomorphisms and suppose that $X(\phi) = X(\psi)$.  Then $\phi_1 \leafprod \phi_2 \leafprod \phi_1 \leafprod \phi_2 \leafprod \id_M = \psi_1 \leafprod \psi_2 \leafprod \psi_1 \leafprod \psi_2 \leafprod \id_M$ where each $\phi_i = \restr{\phi}{P_i}$ and each $\psi_i = \restr{\psi}{P_i}$.  Therefore, each $\phi_i = \psi_i$ so $X_{(P_1, S_2, \bmg), (Q_1, S_2', \bmh)}$ is injective.
\end{proof}

Correctness of our reduction now follows.

\begin{corollary}
  \label{cor:aug-alpha-red-cor}
  Let $(P_1, S_2, \bmg)$ and $(Q_1, S_2', \bmh)$ be augmented $\alpha$-composition pairs for the solvable groups $G$ and $H$.  Then $(P_1, S_2, \bmg) \cong (Q_1, S_2', \bmh)$ \ifft $X(P_1, S_2, \bmg) \cong X(Q_1, S_2', \bmh)$.
\end{corollary}

Because $X$ is defined in terms of leaf products of structures that can be computed in polynomial time, it is immediate that $X$ can also be evaluated in polynomial time.

\begin{lemma}
  \label{lem:X-poly}
  Let $(P_1, S_2, \bmg)$ and $(Q_1, S_2', \bmh)$ be augmented $\alpha$-composition pairs for the solvable groups $G$ and $H$ and let $\phi : (P_1, S_2, \bmg) \ra (Q_1, S_2', \bmh)$ be an isomorphism.  Then both $X(P_1, S_2, \bmg)$ and $X(\phi)$ can be computed in polynomial time.
\end{lemma}

The last ingredient that we require for our algorithm for augmented $\alpha$-composition pair isomorphism is a bound on the degree of the graph.

\begin{lemma}
  \label{lem:aug-alpha-graph}
  Let $(P_1, S_2, \bmg)$ be an augmented $\alpha$-composition pair for the solvable group $G$.  Then the graph $X(P_1, S_2, \bmg)$ has degree at most $\max\{\alpha + 1, 4\}$ and size $O(n^2)$.
\end{lemma}

\begin{proof}
  The trees $T(P_1, \bmg)$, $T(S_2)$ and $M$ have degrees $3$, at most $\alpha + 1$ and $3$ respectively.  Since $\abs{P_1} \abs{P_2} = n$, the size of $T(P_1, \bmg) \leafprod T(S_2)$ is $O(n)$.  Thus, $T(P_1, \bmg) \leafprod T(S_2) \leafprod T(P_1, \bmg) \leafprod T(S_2) \leafprod M$ has size $O(n^2)$ and degree at most $\max\{\alpha + 1, 4\}$.
\end{proof}

Finally, we obtain our result for augmented $\alpha$-composition pair isomorphism.

\begin{theorem}
  \label{thm:aug-alpha-iso}
  Let $(P_1, S_2, \bmg)$ and $(Q_1, S_2', \bmh)$ be augmented $\alpha$-composition pairs for the solvable groups $G$ and $H$.  Then we can test if $(P_1, S_2, \bmg) \cong (Q_1, S_2', \bmh)$ in $n^{O(\alpha)}$ time.
\end{theorem}

\begin{proof}
  By \lemref{X-poly}, we can compute the graphs $X(P_1, S_2, \bmg)$ and $X(Q_1, S_2', \bmh)$ in polynomial time.  By \lemref{aug-alpha-graph} and \thmref{const-deg-can}, we can decide if $X(P_1, S_2, \bmg) \cong X(Q_1, S_2', \bmh)$ in $n^{O(\alpha)}$ time.  Finally, \corref{aug-alpha-red-cor} tells us that $(P_1, S_2, \bmg) \cong (Q_1, S_2', \bmh)$ \ifft $X(P_1, S_2, \bmg) \cong X(Q_1, S_2', \bmh)$.
\end{proof}

Using \lemref{alpha-comp-red}, we obtain the following corollary.

\begin{corollary}
  \label{cor:alpha-iso}
  Let $(P_1, S_2)$ and $(Q_1, S_2')$ be $\alpha$-composition pairs for the solvable groups $G$ and $H$.  Then we can test if $(P_1, S_2) \cong (Q_1, S_2')$ in $n^{O(\alpha) + \log_{\alpha} n}$ time.
\end{corollary}

\subsection{Canonization}
In this subsection, we extend our results for testing isomorphism of $\alpha$-composition pairs to canonization.  This result can be leveraged to obtain faster algorithms for solvable-group isomorphism via collision arguments~\cite{rosenbaum2013b}.  Our canonization algorithm requires another \inpuborpriv{map}{functor} $Y$ that reverses the action of $X$ by sending back to the augmented $\alpha$-composition pairs from which they arise.  We start with the definition for $Y$.  As with $X$, we overload notation so that $Y$ can also be applied to isomorphisms between graphs.

\begin{definition}
  \label{defn:Y-A}
  For each augmented $\alpha$-composition pair $(P_1, S_2, \bmg)$ for a solvable group $G$ and each graph $A \cong X(P_1, S_2, \bmg)$, we fix an arbitrary isomorphism $\pi : X(P_1, S_2, \bmg) \ra A$.  Let $P_{2, 0} = 1 \tril \cdots \tril P_{2, m} = P_2$ be the subgroup chain for $S_2$.  Then we define $Y(A) = (\pi[P_1 \times \{1\}], \pi[\{1\} \times P_{2, 0}] \tril \cdots \tril \pi[\{1\} \times P_{2, m}], \pi(\bmg))$.

  Here, $\pi[\setb{(x_1, x_2)}{x_i \in P_i}]$ is interpreted as a group containing each $\pi[\{1\} \times P_{2, i}]$ as a subgroup.  For each $x_i, y_i, z_i \in P_i$, we define $\pi(x_1, x_2) \pi(y_1, y_2) = \pi(z_1, z_2)$ \ifft there exists a path $(a_{\pi(x)} a_{\pi(y)}, a_{\pi(z)})$ colored $(\text{``left''}, \text{``right''}, \text{``equals''})$, such that $a_{\pi(x)}$, $a_{\pi(y)}$ and $a_{\pi(z)}$ are descendants of the nodes $\pi(x_1, x_2)$, $\pi(y_1, y_2)$ and $\pi(z_1, z_2)$ in the image of the tree $T(P_1, \bmg) \leafprod T(S_2) \leafprod T(P_1, \bmg) \leafprod T(S_2) \leafprod M$ under $\pi$.

  Let $(P_1, S_2, \bmg)$ and $(Q_1, S_2', \bmh)$ be augmented $\alpha$-composition pairs for the groups $G$ and $H$ and consider the graphs $A \cong X(P_1, S_2, \bmg)$ and  $A' \cong X(Q_1, S_2', \bmh)$.  Let
  $\pi : X(P_1, S_2, \bmg) \ra A$ and $\pi' : X(Q_1, S_2', \bmh) \ra A'$ be the fixed isomorphisms chosen above.  Then for each isomorphism $\theta : A \ra A'$, we define $Y(\theta) : \pi[\setb{(x_1, x_2)}{x_i \in P_i}] \ra \pi'[\setb{(x_1, x_2)}{x_i \in Q_i}]$ to be $\restr{\theta}{\pi[\setb{(x_1, x_2)}{x_i \in P_i}]}$.
\end{definition}

As for $X$, we define $Y_{A, A'} : \iso(A, A') \ra \iso(Y(A), Y(A'))$ by $\theta \mapsto Y(\theta)$ for each pair of graphs $A, A' \in \augcomptree$.

Our first step is to show that $Y$ is well-defined.  Once this is proved, we can leverage \thmref{X-fff} to show that each $Y_{A, A'}$ is bijective.  This allows us to define a canonical form for augmented $\alpha$-composition pairs in terms of $\can_{\graph}$, $X$ and $Y$.

\begin{lemma}
  \label{lem:Y-well-def}
  Let $(P_1, S_2, \bmg)$ be an augmented $\alpha$-composition pair for the solvable group $G$, let $A$ be a graph and let $\pi : X(P_1, S_2, \bmg) \ra A$ be an isomorphism.  Then $Y(A)$ is a well-defined augmented composition pair and can be computed in polynomial time.  Moreover, $Y(\pi) : (P_1, S_2, \bmg) \ra Y(A)$ is an isomorphism.
\end{lemma}

\begin{proof}
  We claim that $\pi[\setb{(x_1, x_2)}{x_i \in P_i}]$ is indeed a group if interpreted according to~\defref{Y-A}.  Let $x_i, y_i, z_i \in P_i$.  Then $\pi(x_1, x_2) \pi(y_1, y_2) = \pi(z_1, z_2)$ \ifft there exists a path $(a_{\pi(x)} a_{\pi(y)}, a_{\pi(z)})$ colored $(\text{``left''}, \text{``right''}, \text{``equals''})$, such that $a_{\pi(x)}$, $a_{\pi(y)}$ and $a_{\pi(z)}$ are descendants of the nodes $\pi(x_1, x_2)$, $\pi(y_1, y_2)$ and $\pi(z_1, z_2)$ in $A$.  Since $\pi$ is an isomorphism, this is equivalent to the existence of a path $(a_{x} a_{y}, a_{z})$ colored $(\text{``left''}, \text{``right''}, \text{``equals''})$, such that $a_x$, $a_y$ and $a_z$ are descendants of the nodes $(x_1, x_2)$, $(y_1, y_2)$ and $(z_1, z_2)$ in $X(P_1, S_2, \bmg)$.

  This is in turn equivalent to the existence of a $W$-sequence from $x$ to $y$ to $z$ where $x = x_1 x_2$, $y = y_1 y_2$ and $z = z_1 z_2$.  By definition, this $W$-sequence exists \ifft $x y = z$.  Therefore, $\pi[\setb{(x_1, x_2)}{x_i \in P_i}]$ is a group and $Y(\pi)$ is a group isomorphism from $G$ to $\pi[\setb{(x_1, x_2)}{x_i \in P_i}]$.  It is immediate that $Y(A)$ is an augmented $\alpha$-composition pair and $Y(\pi)$ is an augmented $\alpha$-composition pair isomorphism.

  Now we show how to compute $Y(A)$ in polynomial time.  Let $\ell = \lceil \log \abs{P_1} \rceil$ and let the subgroup chain for $S_2$ be $P_{2, 0} = 1 \tril \cdots \tril P_{2, m}$.  Then $\ell$ is the height of $T(P_1, \bmg)$ and $m$ is the height of $T(S_2)$.  Thus, by~\defref{X-PS}, $\pi[P_1 \times \{1\}]$ consists of the nodes in $A$ colored ``second identity'' at a depth of $\ell + m$ from the root.

  To compute each $\pi[\{1\} \times P_{2, k}]$, we first find the node $\pi(1, 1)$; this is the identity element of the group $\pi[\setb{(x_1, x_2)}{x_i \in P_i}]$.  The node $\pi(1, P_{2, k})$ is the node on the path from the root to $\pi(1, 1)$ in $A$ that is at a distance of $\ell + k$ from the root.  Then, by~\defref{X-PS}, each $\pi[\{1\} \times P_{2, k}]$ consists of the nodes in $A$ descended from $\pi(1, P_{2, k})$ that are at a distance of $m - k$ from $\pi(1, P_2)$.
\end{proof}

Now we can show that \inpuborpriv{each $Y_{A, A'}$ is surjective}{$Y$ is a full functor}.

\begin{theorem}
  \label{thm:Y-fff}
  \inpuborpriv{Consider the graphs $A, A' \in \augcomptree$.  Then $Y_{A, A'}$ is a bijection and both $Y(A)$ and $Y(\theta)$ where $\theta \in \iso(Y(A), Y(A'))$ can be computed in polynomial time}{$Y : \augcomptree \ra \augcomp$ is a fully faithful functor and can be evaluated in polynomial time}.
\end{theorem}

\begin{proof}
  Let $(P_1, S_2, \bmg)$ and $(Q_1, S_2', \bmh)$ be augmented $\alpha$-composition pairs for the solvable groups $G$ and $H$ such that $\pi : X(P_1, S_2, \bmg) \ra A$, $\pi' : X(Q_1, S_2', \bmh) \ra A'$ and $\theta : A \ra A'$ are isomorphisms.
  
  First, we observe that $Y$ respects composition\inpriv{ and the identity} and let $\psi = \theta \pi : X(P_1, S_2, \bmg) \ra A'$.  Since $\theta$ and $\pi$ are isomorphisms so is $\psi$; \lemref{Y-well-def} then implies that $Y(\psi) = Y(\theta) Y(\pi)$ is also an isomorphism.  Therefore, $Y(\theta) = Y(\psi) (Y(\pi))^{-1}$ is an isomorphism and so \inpuborpriv{$Y_{A, A'}$ is a well-defined function}{$Y$ is a well-defined functor}.

  Now we prove that \inpuborpriv{$Y_{A, A'}$ is a bijection}{$Y$ is fully faithful}.  It follows from Definitions~\ref{defn:X-PS} and~\ref{defn:Y-A} that \inpuborpriv{$Y X = I_{\augcomp}$}{$Y_{X_{(P_1, S_2, \bmg)}, X_{(Q_1, S_2', \bmh)}} X_{(P_1, S_2, \bmg), (Q_1, S_2', \bmh)} = I_{\augcomp}$}.  By \thmref{X-fff}, $X_{(P_1, S_2, \bmg), (Q_1, S_2', \bmh)}$ is bijective; this implies that $Y_{X(P_1, S_2, \bmg), X(Q_1, S_2', \bmh)}$ is also bijective since the identity is bijective.  Now we just need to show that $Y_{A, A'}$ is bijective.  For each isomorphism $\theta : A \ra A'$, there exists an isomorphism $\rho : X(P_1, S_2, \bmg) \ra X(Q_1, S_2', \bmh)$ such that $\theta = \pi' \rho \pi^{-1}$.  It follows that $Y(\theta) = Y(\pi') Y(\rho) Y(\pi^{-1})$ from which we see that $Y_{A, A'}$ is indeed bijective.\inpriv{  Thus, $Y$ is fully faithful.}

  We already showed that $Y(A)$ can be computed in polynomial time in \lemref{Y-well-def} and it follows easily from \defref{Y-A} that $Y(\theta)$ can be computed in polynomial time.
\end{proof}

While \thmref{Y-fff} is enough to obtain our canonization results, we point out that $X$ and $Y$ form a category equivalence \inpuborpriv{when viewed as functors}{(see \appref{X-Y-equiv})}.\inpub{  Moreover, the results of this section can be derived from this more general fact.}

To construct our canonical form for augmented $\alpha$-composition pairs, we convert our augmented $\alpha$-composition pairs to graphs of degree at most $\alpha + O(1)$ by applying $X$.  Then we compute the canonical form of the resulting graph using \thmref{const-deg-can} and convert it back into an augmented $\alpha$-composition pair by applying $Y$.  We use $\can_{\graph}$ to denote the map from graphs to their canonical forms from \thmref{const-deg-can}.

\begin{theorem}
  \label{thm:aug-alpha-can}
  $Y \circ \can_{\graph} \circ X$ is a canonical form for augmented $\alpha$-composition pairs.  Moreover, for any $\alpha$-composition pair $(P_1, S_2, \bmg)$, we can compute $(Y \circ \can_{\graph} \circ X)(P_1, S_2, \bmg)$ in $n^{O(\alpha)}$ time.
\end{theorem}

\begin{proof}
  Consider two $\alpha$-composition pairs $(P_1, S_2, \bmg)$ and $(Q_1, S_2', \bmh)$ for the solvable groups $G$ and $H$.  By \corref{aug-alpha-red-cor}, $(P_1, S_2, \bmg) \cong (Q_1, S_2', \bmh)$ \ifft 
  \begin{equation*}
    X(P_1, S_2, \bmg) \cong X(Q_1, S_2', \bmh).
  \end{equation*}
Thus, $(P_1, S_2, \bmg) \cong (Q_1, S_2', \bmh)$ \ifft
\begin{equation*}
  \can_{\graph}(X(P_1, S_2, \bmg)) = \can_{\graph}(X(Q_1, S_2', \bmh))
\end{equation*}
Now, clearly, if $(P_1, S_2, \bmg) \cong (Q_1, S_2', \bmh)$,
\begin{equation*}
  Y(\can_{\graph}(X(P_1, S_2, \bmg))) = Y(\can_{\graph}(X(Q_1, S_2', \bmh)))
\end{equation*}
On the other hand, if $(P_1, S_2, \bmg) \not\cong (Q_1, S_2', \bmh)$, then
\begin{align*}
  \can_{\graph}(X(P_1, S_2, \bmg))    & \not\cong \can_{\graph}(X(Q_1, S_2', \bmh)) \\
  Y(\can_{\graph}(X(P_1, S_2, \bmg))) & \not\cong Y(\can_{\graph}(X(Q_1, S_2', \bmh))) \\
  Y(\can_{\graph}(X(P_1, S_2, \bmg))) & \not= Y(\can_{\graph}(X(Q_1, S_2', \bmh))).
\end{align*}
 Thus, $Y \circ \can_{\graph} \circ X$ is a complete invariant.  Also, $X(P_1, S_2, \bmg) \cong \can_{\graph}(X(P_1, S_2, \bmg))$ so since $Y X = I_{\augcomp}$, we have $(P_1, S_2, \bmg) \cong Y(\can_{\graph}(X(P_1, S_2, \bmg)))$ by \thmref{Y-fff}.  Thus, $Y \circ \can_{\graph} \circ X$ is a canonical form.

  Lastly, we show that $Y(\can_{\graph}(X(P_1, S_2, \bmg)))$ can be computed in $n^{O(\alpha)}$ time.  By \thmref{X-fff}, we can compute $X(P_1, S_2, \bmg)$ in polynomial time.  By \lemref{aug-alpha-graph} and \thmref{const-deg-can}, it takes $n^{O(\alpha)}$ time to compute $\can_{\graph}(X(P_1, S_2, \bmg))$.  Finally, by \thmref{Y-fff}, we can compute $Y(\can_{\graph}(X(P_1, S_2, \bmg)))$ in polynomial time from $\can_{\graph}(X(P_1, S_2, \bmg))$.
\end{proof}

\section{Algorithms for solvable-group isomorphism and canonization}
\label{sec:sol-algorithms}
Armed with the results of Sections~\ref{sec:alpha-red} and \ref{sec:graph-red}, it is easy to prove \thmref{sol-iso} as promised in the introduction.

\soliso*

\begin{proof}
  Let $\alpha$ be a parameter to be chosen later.  By Lemmas~\ref{lem:sol-red} and~\ref{lem:alpha-red}, we can reduce solvable-group isomorphism to $\alpha$-composition pair isomorphism in $n^{(1 / 2 \log_p n + O(1))}$ where $p$ is the smallest prime dividing the order of the group.  Now, by \lemref{alpha-comp-red}, $\alpha$-composition pair isomorphism reduces to augmented $\alpha$-composition pair isomorphism in $n^{\log_{\alpha} n + O(1)}$.  Thus, we can reduce solvable-group isomorphism to augmented $\alpha$-composition pair isomorphism in $n^{(1 / 2) \log_p n + \log_{\alpha} n + O(1)}$ time.
  
  \begin{sloppypar}
    Applying \thmref{aug-alpha-iso}, we obtain an $n^{(1 / 2) \log_p n + \log_{\alpha} n + O(\alpha)}$ time algorithm for solvable-group isomorphism.  The optimal choice for $\alpha$ is $\alpha = \log n / \log \log n$.  The complexity is then $n^{(1 / 2) \log_p n + O(\log n / \log \log n)}$ as claimed.
  \end{sloppypar}
\end{proof}

Our algorithm for solvable-group canonization follows by a similar argument.

\begin{theorem}
  \label{thm:sol-group-can}
  Solvable-group canonization is in $n^{(1 / 2) \log_p n + O(\log n / \log \log n)}$ deterministic time.
\end{theorem}

\begin{proof}
  Let $\alpha$ be a parameter to be chosen later.  By Lemmas~\ref{lem:group-red-alpha-decomp-can} and~\ref{lem:alpha-decomp-red-pair-can}, we can reduce solvable-group canonization to $\alpha$-composition pair canonization in $n^{(1 / 2 \log_p n + O(1))}$ where $p$ is the smallest prime dividing the order of the group.  $\alpha$-composition pair canonization in turn reduces to augmented $\alpha$-composition pair canonization in $n^{\log_{\alpha} n + O(1)}$ time by enumerating all possible choices of ordered generating sets and choosing the canonical form that comes first lexicographically.  Thus, we can reduce solvable-group canonization to augmented $\alpha$-composition pair canonization in $n^{(1 / 2) \log_p n + \log_{\alpha} n + O(1)}$ time.
  
  Applying \thmref{aug-alpha-can}, we obtain an $n^{(1 / 2) \log_p n + \log_{\alpha} n + O(\alpha))}$ time algorithm for solvable-group canonization.  The optimal choice for $\alpha$ is $\alpha = \log n / \log \log n$.  The complexity is again $n^{(1 / 2) \log_p n + O(\log n / \log \log n)}$ as claimed.
\end{proof}


\section*{Acknowledgements}
I thank Laci Babai for suggesting the simplification of using the concept of a $\alpha$-decomposition and other comments, Paul Beame and Aram Harrow for useful discussions and feedback, Joshua Grochow for additional references and the anonymous reviewers for helpful comments.  Richard Lipton asked if the techniques used for $p$-groups could be applied to groups of order $2^a p^b$ where $p$ is an odd prime.  This inspired our more efficient algorithm for solvable-group isomorphism.  Part of this work was completed at the Center for Theoretical Physics at the Massachusetts Institute of Technology.  This research was funded by the DoD AFOSR through an NDSEG fellowship, by the Simons Foundation through a Simons Award for Graduate Students in Theoretical Computer Science (grant \#316172) and by the NSF under grants CCF-0916400 and CCF-1111382.

\inpriv{\inlong{
    \appendix
    \newpage
    \section{$X$ and $Y$ are a category equivalence}
    \label{app:X-Y-equiv}
    \begin{theorem}
  The functors $X : \augcomp \ra \augcomptree$ and $Y : \augcomptree \ra \augcomp$ form a category equivalence.
\end{theorem}

\begin{proof}
  We start by noting that $YX = I_{\augcomp}$ so $YX$ is naturally isomorphic to the identity functor.  Let $A \in \augcomptree$ be a graph and let $S$ be a augcomposition series such that $A \cong X(S)$.  Then $XY(A) \cong XYX(S)$ so $XY(A) \cong X(S)$ and $A \cong XY(A)$.  Then for each graph $A \in \augcomptree$, choose an isomorphism $\pi_A : XY(A) \ra A$.  We note that $XYXY = XY$ so that $XY(\pi) : XY(A) \ra XY(A)$ is an automorphism.  Define $\eta_A = \pi_A \circ XY(\pi^{-1})$; we claim that $\eta : XY \ra I_{\augcomptree}$ is a natural isomorphism.  Clearly, each $\eta_A$ is an isomorphism.  We need to show that for each $A, A' \in \augcomptree$ and every isomorphism $\theta : A \ra A'$, $\eta_{A'} \circ XY(\theta) = \theta \circ \eta_A$.  Since $XY$ is fully faithful by Theorems~\ref{thm:X-fff} and~\ref{thm:Y-fff}, this is equivalent to $XY(\eta_{A'}) \circ XYXY(\theta) = XY(\theta) \circ XY(\eta_A)$ which holds since $XY(\eta_A) = \id_{XY(A)}$.
\end{proof}
}}

\inlong{\newpage}
\bibliographystyle{initials}
\bibliography{$HOME/LaTeX/computer-science-references,$HOME/LaTeX/math-references,$HOME/LaTeX/quantum-computing-references} 

\end{document}